\documentclass[aps,prl,10pt,twocolumn,nofootinbib]{revtex4-2}
\usepackage{amsfonts,amssymb,amsmath}
\usepackage{amsthm}
\usepackage{mathrsfs}
\usepackage[hidelinks]{hyperref}
\usepackage{booktabs}
\newtheorem{theorem}{Theorem}

\newcommand{\F}{\mathbb F}
\newcommand{\M}{\mathbb M}

\begin{document}
\title{Four-Point String Monodromy from Six-Point Supersymmetric Consistency}
\author{Jianqi Sheng}
\affiliation{Department of Physics, City University of Hong Kong, Hong Kong}
\email{sheng.jq@cityu.edu.hk}

\begin{abstract}
Can finite-particle consistency determine an infinite string amplitude?
For four-dimensional planar tree-level $\mathcal N=4$ super Yang--Mills
effective theories with a specified scalar parity condition, we derive
all-order constraints directly from physical six-point Ward identities and
local factorization. An overlapping-channel limit forces a scalar cocycle,
while an explicit differential Ward functional eliminates all relevant five-
and six-point contact freedom and fixes the four-point kernel by
$q''=-6c_{00}q$. Its normalized branches are linear, trigonometric, and
hyperbolic, and the target amplitude obeys the associated four-point
monodromy. If its entire pole-subtracted axis sequence additionally admits
one common nonzero positive moment measure, the trigonometric branch is
selected and the complete Veneziano low-energy germ is unique up to scale.
The result separates a formal low-energy necessity theorem from its
conditional analytic-germ uniqueness corollary; no worldsheet or auxiliary
KLT-module membership is assumed for the physical theorem.
\end{abstract}
\maketitle

The Veneziano amplitude packages infinitely many low-energy couplings and
massive poles into a single function \cite{Veneziano:1968yb}. Its relations
between color orderings are usually explained by worldsheet monodromy
\cite{Plahte:1970wy,BjerrumBohr:2010hn}. A complementary question is whether
consistency among a finite number of external particles can already force
those relations at every order in energy. Agreement of finitely many Wilson
coefficients cannot answer this question: a deformation may first appear
beyond the highest order computed.

There are two relevant routes. In the KLT bootstrap, doubly ordered
amplitudes organize deformations of ordering relations
\cite{Chi:2021mio,Chen:2022shl,Chen:2023dcx}. In a physical supersymmetric
EFT, Elvang, Herderschee, and Morales derived restrictive finite-order
nonlinear constraints from six-scalar Ward identities with a scalar parity
assumption, and found numerical evidence for Veneziano uniqueness after
adding ultraviolet assumptions \cite{EHM:2026}. Separately, Wan and Zhou
proved analytic uniqueness from dispersive positivity supplemented by
monodromy or splitting input \cite{Wan:2026}. Huang, Wan, Wang, and Zhou derive all-order gauge-theory constraints, monodromy, and conditional Veneziano uniqueness from a three-channel collinear master equation \cite{Huang:2026infrared}. Here we construct an explicit differential Ward functional that annihilates the complete five- and six-point local freedom, with polynomial certificates at arbitrary EFT order. The SM relates the resulting kernel equation directly to their gauge-theory master equation.

\emph{Setting and result.---}Consider a four-dimensional, planar single-trace,
tree-level local $\mathcal N=4$ SYM EFT with unbroken $SU(4)$ R symmetry.
The SYM three-point coupling $g$ is nonzero; massless factorization channels
sum over the specified SYM multiplet. We require standard planar poles,
exclude spurious poles and simultaneous propagators in incompatible
channels, and allow arbitrary local contact polynomials at each EFT order.
In addition, the component
$Z_1=A_6[z^{12}z^{13}z^{14}z^{23}z^{34}z^{24}]/g^4$ and its cyclic images
are invariant under exchanging angle and square spinors. This scalar parity
condition is an independent assumption, not a consequence of supersymmetry
or $SU(4)$ alone \cite{EHM:2026}.

Divide the stripped four-point function by $g^2$ and denote its two planar
channels by $s,t$:
\begin{equation}
 A(s,t)=-\frac1{st}+f(s,t),\quad f(s,t)=f(t,s),\quad B=1-stf.
 \label{eq:phys-normalization}
\end{equation}
The physical theorem is the pair of formal identities
\begin{align}
 B(x,y)B(x+y,z)&=B(y,z)B(x,y+z),
 \label{eq:phys-cocycle}\\
 zH'(z)+3H(z)-z^2H(z)^2&=6c_{00},
 \label{eq:phys-riccati}
\end{align}
where $H(z)=f(z,0)+f(-z,0)$ and $c_{00}=f(0,0)$.
Consequently, the kernel defined by the target itself satisfies
\begin{equation}
 \begin{gathered}
 q_A(z)=\frac{z}{B(z,-z)},\qquad q_A''=-6c_{00}q_A,\\
 q_A(0)=0,\qquad q_A'(0)=1.
 \end{gathered}
 \label{eq:phys-ode}
\end{equation}
These are necessary conditions for a local completion through six points.
No positivity or UV completion is used in deriving them. The proof is
coefficientwise in the EFT expansion and does not require convergence at
the nonzero kinematic points used below. Detailed constructions and exact
certificates are given in the Supplemental Material (SM) \cite{supp}.

\emph{Overlapping channels force the cocycle.---}Use adjacent coordinates
$a=s_{12}$, $b=s_{23}$, $c=s_{34}$, $d=s_{45}$, $e=s_{56}$,
$f_6=s_{61}$ and three-particle channels $U=s_{123}$, $V=s_{234}$,
$W=s_{345}$. On the physical subspace $p_6=p_2$,
$f_6=a$, $W=2a$, $c=a+e+V$, and $d=a+b+U$.
The two parity-related Ward rows imply
\begin{equation}
\begin{split}
0={}&b[(a+b)V-(a+e)U]Z_2\\
&-a[(a+b)V+bU]Z_3+a(a+b+e)UZ_4.
\end{split}
\label{eq:phys-overlap-row}
\end{equation}
Take $U,V\to0$ along the four-dimensional spinor path supplied in the SM,
while all adjacent channels remain nonzero. After subtracting the simple
three-particle poles, the remaining terms are regular. Their coefficients
in Eq.~\eqref{eq:phys-overlap-row} vanish in the limit. The two surviving
residue products have opposite, generically nonzero coefficients, giving
\begin{equation}
 B(a,b)B(a+b,e)=B(b,a+e)B(e,a).
 \label{eq:phys-overlap-result}
\end{equation}
The path maps dominantly onto $(a,b,e)$, so this is the unrestricted
cocycle \eqref{eq:phys-cocycle}, rather than an identity confined to a
special four-point slice. Incompatible channels never require a double
residue. Normalization then gives
\begin{equation}
 \begin{gathered}
 B(s,t)=\frac{\gamma(s)\gamma(t)}{\gamma(s+t)},\\
 (\log\gamma)'(z)=z f(z,0),\qquad \gamma(0)=1,
 \end{gathered}
 \label{eq:phys-cochain}
\end{equation}
where the irrelevant linear term in $\log\gamma$ is set to zero.

\emph{A Ward functional closes the kernel at every order.---}The cocycle
leaves one arbitrary cochain. To constrain its even part, the two-particle
factorizations, including all five-point contact freedom, must be retained.
Expand the odd Ward row in a five-epsilon basis and include its full Gram
relations. The local six-point Ward problem becomes
$\mathbb F\in\operatorname{im}\mathbb M$, with a fixed $6\times12$
polynomial matrix. Its first six columns describe contact terms; the
others impose the scalar Gram and Gram--epsilon relations. Thus transverse
derivatives are taken in a polynomial presentation that respects physical
four-dimensional kinematics.

On the curve
\begin{equation}
 \Gamma(\tau)=(\tau,0,-\tau,1,0,-1,2\tau,-2\tau,0),
 \label{eq:phys-curve}
\end{equation}
we construct an explicit second-order differential functional $J_\tau$.
Its rational weights, given in the SM, obey the polynomial identity
$J_\tau[\mathbb M C]=0$ for every local multiplier $C$.
The complete five-point freedom is generated by one arbitrary polynomial
$T$ and one cyclic polynomial $P$. Its response reduces exactly to
\begin{equation}
\begin{split}
\delta J_\tau=K(\tau)\big[&T(0,-1,2\tau,-2\tau,1)\\
&-T(-2\tau,1,0,-1,2\tau)\big],
\end{split}
\label{eq:phys-five-response}
\end{equation}
where $K(\tau)=-32(\tau-1)(2\tau-1)^2(3\tau^2-3\tau+1)$ and $K(0)=32$.
All derivatives of $T$ and all $P$ contributions cancel. At four-point
degree $N$, $T$ has degree $N-3$. Therefore
\begin{equation}
 \mathcal L_N=4[\tau^{N-2}]\,J_\tau/K(\tau),\qquad N\ge2\ \text{even},
 \label{eq:phys-projector}
\end{equation}
annihilates every six-point contact and every relevant five-point
deformation. The $N=2$ sector contains no independent five-point freedom.

Apply this same functional to the canonical factorization source. The SM
derives its response for an arbitrary symmetric function $f$, without an
EFT truncation. Coefficient extraction gives
$\mathcal L_N[\mathbb F_N]=[z^N]\mathcal P(z)$, where
\begin{equation}
\begin{split}
\mathcal P(z)={}&2f(2z,-2z)-\tfrac12[f(2z,-z)+f(z,-2z)]\\
&-H(2z)+\tfrac12H(z)-z^2f(2z,-z)f(z,-2z).
\end{split}
\label{eq:phys-source}
\end{equation}
Physical solvability forces all its even coefficients to vanish; its odd
coefficients and constant term vanish identically. With
$D(z)=B(z,-z)$, the cocycle reduces $\mathcal P=0$ to
\begin{equation}
 D(z)-1=z^2[H(z)-\tfrac12H(z/2)].
 \label{eq:phys-dyadic}
\end{equation}
Since $D'/D=zH$, this is equivalent to the doubling identity
\[
[1+q_A'(2z)]q_A(z)=q_A(2z)q_A'(z).
\]
Expand $q_A=z+\sum_{n\ge1}b_{2n+1}z^{2n+1}$. The multiplier of the new coefficient
at step $n$ is $(2n-1)4^n-4n$. It vanishes for $n=1$ and is nonzero for
$n\ge2$. The cubic coefficient $b_3=-c_{00}$ therefore fixes the entire
series. The normalized solution of Eq.~\eqref{eq:phys-ode} satisfies the
same doubling identity and cubic datum, proving
Eqs.~\eqref{eq:phys-riccati}--\eqref{eq:phys-ode}.

\emph{Physical monodromy and residual freedom.---}For real $c_{00}$ the
kernel is linear, trigonometric, or hyperbolic. Equivalently, write
$q_A(z)=\sin(\lambda z)/\lambda$ with $\lambda^2=6c_{00}$, including its
zero and imaginary limits. For $s+t+u=0$, Eq.~\eqref{eq:phys-cochain} gives
$(A(s,t),A(s,u),A(t,u))=\Omega(q_A(u),q_A(t),q_A(s))$, where
$\Omega=-\gamma(s)\gamma(t)\gamma(u)/(stu)$. The sine addition law implies
\begin{equation}
 A(s,t)+e^{i\lambda s}A(s,u)+e^{-i\lambda t}A(t,u)=0.
 \label{eq:phys-monodromy}
\end{equation}
Thus the physical target supplies its own four-point monodromy, without
an assumed entrance into an auxiliary KLT module.
The two necessary conditions leave
\begin{equation}
 \gamma(z)=\sqrt{\frac{\lambda z}{\sin(\lambda z)}}e^{\phi(z)},
 \qquad \phi\in z^3\mathbb R[[z^2]].
 \label{eq:phys-odd-freedom}
\end{equation}
This describes four-point data satisfying those conditions, not a
construction of a multipoint completion for every $\phi$.

\emph{One positive axis measure fixes the amplitude.---}Let
$f(s,0)=\sum_{k\ge0}a_k s^k$. Add the explicit UV premise
\begin{equation}
 a_k=\int_0^\infty x^k\,d\mu(x),\qquad k\ge0,
 \label{eq:phys-positive-moments}
\end{equation}
for one finite nonnegative measure, with all moments finite. A justified
unsubtracted positive spectral representation of the pole-subtracted axis
implies this condition. Massless tree factorization alone does not; an
arbitrary subtraction polynomial cannot be added to it.

If $\mu=0$, the cocycle gives pure SYM. Otherwise $a_0>0$, and
$\alpha=\sqrt{6a_0}/\pi$ is real and positive. The Riccati equation fixes
$a_{2n}=\alpha^{2n+2}\zeta(2n+2)$. For any $L>\alpha$,
$\mu([L,\infty))\le a_{2n}/L^{2n}\to0$, so the measure is supported on
$[0,\alpha]$. Its even moments match those of
\begin{equation}
 \mu_* =\sum_{m\ge1}\frac{\alpha^2}{m^2}\delta_{\alpha/m}.
 \label{eq:phys-discrete-measure}
\end{equation}
Pushing both measures forward by $x\mapsto x^2$ gives equal integer
moments on a compact interval. Polynomial density implies equality of the
pushforwards, and injectivity on the nonnegative axis gives $\mu=\mu_*$.
Thus $a_k=\alpha^{k+2}\zeta(k+2)$ at every order. This use of moment
determinacy to recover odd axis data is standard and appears in
Ref.~\cite{Wan:2026}; the physical cocycle here then fixes all mixed
coefficients without additional angular-derivative moment sequences.

Integrating Eq.~\eqref{eq:phys-cochain} yields
$\gamma(z)=e^{-\gamma_E\alpha z}\Gamma(1-\alpha z)$ and hence
\begin{equation}
 \boxed{A(s,t)=-\alpha^2
 \frac{\Gamma(-\alpha s)\Gamma(-\alpha t)}
 {\Gamma(1-\alpha(s+t))}.}
 \label{eq:phys-veneziano}
\end{equation}
The support bound gives $0\le a_k\le a_0\alpha^k$. Hence the axis series and $\log\gamma$ converge for $|z|<1/\alpha$; the cocycle makes the pole-subtracted $f$ holomorphic on $|s|+|t|<1/\alpha$. Thus Eq.~\eqref{eq:phys-veneziano} fixes both the formal series and its analytic realization. It also supplies a mathematical meromorphic continuation. A physical continuation agrees on any common connected meromorphic domain where it exists. If the axis has the stated positive dispersion representation $f(s,0)=\int d\sigma(E)/(E-s)$, its spectral weights are $\alpha/m$ at $E_m=m/\alpha$, with $E$ a mass-squared variable. These are total channel weights.

\emph{Logical boundaries.---}The cocycle is stronger than four-point
monodromy. For example, in fixed units,
$\widetilde A=A_{\rm V}[1+(stu)^2]$ preserves monodromy and the entire
pole-subtracted axis. Its additional factor fails the cocycle, and hence
does not admit the physical completion required here. Conversely,
Eq.~\eqref{eq:phys-odd-freedom} with $\phi=0$ and $c_{00}>0$ has
$a_1=0$, $a_2>0$, which is impossible for a positive measure on
$[0,\infty)$. Neither example is presented as a consistent full S-matrix.

The central result is the elimination of local completion freedom by a
finite differential construction that controls every EFT order. It upgrades
finite-order physical constraints to an exact kernel equation and, with
one specified positive moment measure, to four-point Veneziano uniqueness.
The scalar parity premise and the common positive moment representation specify the assumptions behind these four-point conclusions.

\section*{Data and code availability}
The code and inputs for the symbolic calculations and finite-order checks are openly available in the versioned Zenodo archive~\cite{Sheng:WardCode}. Finite-order and independent spinor implementations provide cross-checks; they are not used to extrapolate the all-order result. No experimental data are used.

\clearpage
\onecolumngrid
\begin{center}\large\bfseries Supplemental Material\end{center}

\setcounter{equation}{0}
\renewcommand{\theequation}{S\arabic{equation}}
\renewcommand{\theHequation}{supp.\arabic{equation}}
\section{Physical cocycle from overlapping channels}
\label{sm:cocycle}
\noindent\textbf{Proposition (necessary physical cocycle).}
Consider a four-dimensional, single-trace, tree-level local EFT with
$\mathcal N=4$ supersymmetry and unbroken $SU(4)$ R symmetry. Assume a
nonzero SYM three-point coupling, standard planar factorization through
only the specified massless SYM multiplet, and no spurious poles. Assume
also that the six-scalar component specified below is invariant under
$\lambda\leftrightarrow\widetilde\lambda$, as are its cyclic images.
Write the normalized stripped four-point function as
\begin{equation}
 A(s,t)=-\frac1{st}+f(s,t),\qquad
 f(s,t)=f(t,s),\qquad
 B(s,t)=-stA(s,t)=1-stf(s,t).
\end{equation}
Then, to every order in the formal local expansion,
\begin{equation}
 B(x,y)B(x+y,z)
 =B(y,z)B(x,y+z).
 \label{eq:physical-cocycle-overlap}
\end{equation}
This is a necessary condition; sufficiency for a full multiparticle
completion is not asserted.

\medskip\noindent\textit{Proof.}
Let $Z_1$ be the normalized component with scalar states
$(12,13,14,23,34,24)$ and let $Z_i$ denote the left cyclic rotations of
these states at fixed color ordering $[123456]$. Define
\begin{equation}
 (a,b,c,d,e,f_6,U,V,W)
 =(s_{12},s_{23},s_{34},s_{45},s_{56},s_{61},s_{123},s_{234},s_{345}).
\end{equation}
Set $x=U-a-b$, $k=V-b-c$, $h=V-e-f_6$, and $u=W-a-f_6$.
The four-point dressings of the three-particle residues are
\begin{equation}
 P_U=B(a,b)B(d,e),\quad
 P_V=B(b,c)B(e,f_6),\quad
 P_W=B(c,d)B(f_6,a).
\end{equation}
It suffices to retain the following three components:
\begin{align}
 Z_2&=-\frac{x}{bU}P_U-\frac{k}{bV}P_V+R_2,\nonumber\\
 Z_3&=-\frac{x}{aU}P_U-\frac{u}{aW}P_W+R_3,\nonumber\\
 Z_4&=-\frac{h}{f_6V}P_V-\frac{u}{f_6W}P_W+R_4.
 \label{eq:physical-cocycle-poles}
\end{align}
The $R_i$ collect all other contributions. Locality will make them
regular on the path below after the displayed poles are subtracted.

Choose parameters $A,B,Y,T$ and put
\begin{equation}
 L=(A+B)Y-B,\quad K=AY^2+B(Y-1)^2,\quad E=-L^2/K,
 \quad m=A+B+E,
\end{equation}
\begin{equation}
 \delta=(T-A-B)/A,\qquad
 v=\frac{-A\delta+m\delta(1-Y)+2E}{L}.
\end{equation}
With the convention $s_{ij}=-\langle ij\rangle[ij]$, take
\begin{equation}
\begin{array}{c|c|c}
 i&\lambda_i&\widetilde\lambda_i\\\hline
 1&(1,0)&(-A,0)\\
 2&(0,1)&(0,1)\\
 3&(1,1)&(-B,\delta)\\
 4&(1,Y)&(m,v)\\
 5&(1,(B-Ym)/(-E))&(-E,-\delta-v)\\
 6&(0,1)&(0,1)
\end{array}
\end{equation}
All legs are null, momentum is conserved, and $p_6=p_2$ identically.
The induced invariants satisfy
\begin{gather}
 a=f_6=A,\quad b=B,\quad e=E,\quad U=T,\quad W=2A,\quad u=0,\nonumber\\
 V=-B(Y-1)^2T/K,\qquad c=A+E+V,\qquad d=A+B+T.
 \label{eq:physical-cocycle-path}
\end{gather}
Work on the nonempty open set where the six adjacent invariants and
$K,L,Y-1,L+A$ do not vanish at $T=0$. The additional zero $s_{26}=0$
is nonplanar and hence is not an allowed pole of any $Z_i$.
The channels $U$ and $V$ overlap, so no simultaneous tree propagator
$1/(UV)$ is present. Thus each $R_i$ in
\eqref{eq:physical-cocycle-poles} is regular at $T=0$.

Let $M(\lambda,\widetilde\lambda)$ be the six-scalar projection of the
five standard $SU(4)$-invariant NMHV superamplitude basis elements,
with the common denominator $[34]^4\langle56\rangle^4$ suppressed.
Direct component projection gives a left null vector
$w(T)=(0,w_2,w_3,w_4,1,0)$, where
\begin{align}
 w_2(T)&=-\frac{B^2(Y-1)[L+A+T(Y-1)]}{(A+B-T)L^2},\nonumber\\
 w_3(T)&=\frac{AB[K+TY(Y-1)]}{(A+B-T)L^2},\nonumber\\
 w_4(T)&=\frac{A(L-TY)}{(A+B-T)L}.
\end{align}
In particular, the ten rational identities
\begin{equation}
 w(T)M(\lambda,\widetilde\lambda)=0,\qquad
 w(-T)M(\widetilde\lambda,\lambda)=0
\end{equation}
hold on the path. The latter $-T$ labels the explicit row formula;
both identities concern the same momentum point, with its spinors
exchanged in the second matrix. The reference denominators remain
nonzero: at $T=0$,
\begin{equation}
 \langle34\rangle=Y-1,\quad [34]=B(L+A)/K,\quad
 \langle56\rangle=1,\quad [56]=L^2/K.
\end{equation}
The parity assumption makes both rows annihilate the same physical
vector $Z$. Their difference is particularly simple:
\begin{equation}
 d(T)=w(T)-w(-T)
 =\frac{2BT}{[(A+B)^2-T^2]L}
 (0,-2B(Y-1),A(2Y-1),-A,0,0).
\end{equation}
Substituting the factorization terms gives the exact identity
\begin{equation}
 0=\frac{2B}{(A+B+T)L}(P_U-P_V)
   +d_2R_2+d_3R_3+d_4R_4.
\end{equation}
Since $d_i=O(T)$ and $R_i$ are regular, the $T\to0$ limit implies
\begin{equation}
 B(A,B)B(A+B,E)
 =B(B,A+E)B(E,A).
\end{equation}
Finally,
\begin{equation}
 \frac{\partial E}{\partial Y}=-\frac{2ABL}{K^2}
\end{equation}
is generically nonzero. The map $(A,B,Y)\mapsto(A,B,E)$ is therefore
dominant. At each homogeneous EFT order the identity is polynomial,
so it extends to unrestricted $(A,B,E)$. Crossing and
$(x,y,z)=(B,A,E)$ give \eqref{eq:physical-cocycle-overlap}.
For a formal series this reasoning is applied after taking each
coefficient of a uniform Mandelstam scaling; no convergence at a
finite kinematic point is required.\hfill$\square$

\paragraph{Coverage and regularity of the physical path.}
The spinors parameterize complexified on-shell kinematics. At each homogeneous EFT order, the Ward and factorization relations are rational identities and therefore extend to this domain. The same repeated-momentum chart with $E$ independent has invariant coordinates $(a,b,e,U,V)=(A,B,E,T,V(A,B,E,T,Y))$; the Jacobian with respect to $(A,B,E,T,Y)$ has the nonzero witness $120/7$ at $(2,3,7,11,2)$. The compact row in the Letter holds on this general chart before $E=-L^2/K$ is imposed. On the special path, $(A,B,Y)=(2,3,2)$ gives adjacent invariants $(2,3,-27/11,5,-49/11,2)$, $W=4$, a limiting cut coefficient $6/35$, and $\partial_YE=-84/121$. These witnesses establish that the required open sets are nonempty; the Ward and residue equations are exact rational identities in all parameters. The substitution $(A,B,T)\mapsto(\rho A,\rho B,\rho T)$ scales all Mandelstams uniformly and justifies coefficientwise use of the path.

\medskip\noindent\textbf{Corollary (the target's own kernel).}
Normalization $B(s,0)=1$ and the cocycle imply
\begin{equation}
 B(s,t)=\frac{\gamma(s)\gamma(t)}{\gamma(s+t)},\qquad
 \gamma(z)=\exp\!\left(\int_0^z x f(x,0)\,dx\right).
\end{equation}
Consequently,
\begin{equation}
 q_A(z)=\frac{z}{\gamma(z)\gamma(-z)}=\frac1{zA(z,-z)}
\end{equation}
is a normalized odd formal function, and for $s+t+u=0$,
\begin{equation}
 q_A(s)A(s,u)
 =-\frac{\gamma(u)}{u\gamma(-s)\gamma(-t)}
 =q_A(t)A(t,u).
\end{equation}
The cocycle step alone does not identify $q_A$ with any independently
constructed auxiliary kernel and by itself does not establish $q_A''=\kappa q_A$; Sec.~\ref{sm:boundary} supplies the additional constraint.
For example, $\gamma(z)=\exp(cz^2/2)$ satisfies the cocycle but gives
$q_A(z)=z\exp(-cz^2)$, for which $q_A''/q_A=-6c+4c^2z^2$.
This example only separates the mathematical conclusions; a full
physical multiparticle completion of that germ is not claimed.

\paragraph{Normalized cochain reconstruction.}
Let $b(s,t)=\log B(s,t)$, which is well defined since $B(0,0)=1$. The cocycle becomes additive. Setting $g(z)=\partial_2b(z,0)=-zf(z,0)$ and differentiating the additive cocycle in its third variable at zero gives
\[
 \partial_2b(s,t)=g(s+t)-g(t).
\]
Define $h(z)=-\int_0^z g(x)\,dx=\int_0^zxf(x,0)\,dx$. Formal integration is defined coefficientwise in characteristic zero. The series $h(s)+h(t)-h(s+t)$ has the same second-variable derivative as $b$ and the same value at $t=0$, so
\[
 B(s,t)=\frac{\gamma(s)\gamma(t)}{\gamma(s+t)},\qquad\gamma=e^h,\qquad\gamma(0)=1,\quad h'(0)=0.
\]
Adding a linear term to $h$ leaves $B$ unchanged; the displayed choice fixes that freedom. Put $D(z)=B(z,-z)=\gamma(z)\gamma(-z)$ and $q(z)=z/D(z)$. Then $D$ is an even unit and
\[
 \frac{D'}D=zH(z),\qquad H(z)=f(z,0)+f(-z,0),\qquad q(z)=z-c_{00}z^3+O(z^5).
\]
Thus $q$ is odd with $q(0)=0$ and $q'(0)=1$.

\section{Physical Ward presentation and complete five-point freedom}
\label{sm:physical-presentation}

This section fixes the polynomial data to which the differential functional
in Sec.~\ref{sm:boundary} is applied. The scalar projection and five-point
MHV basis are those of Ref.~\cite{EHM:2026}, with the normalizations stated
above. We use $f_6=s_{61}$ to distinguish this coordinate from $f(s,t)$.
Let $\rho(a,b,c,d,e,f_6,U,V,W)=(f_6,a,b,c,d,e,W,U,V)$ and set
\[
x=U-a-b,\quad y=U-d-e,\quad u=W-a-f_6,\quad z=W-c-d,
\quad h=V-e-f_6,\quad k=V-b-c.
\]
Two polynomial Ward rows are $r_i=\rho^ir_0$ and
$o_i=(-1)^i\rho^io_0$, $i=0,\ldots,5$, where $r_0=chu(x-y)$.
The odd seed is expanded in
\[
\mathcal E=(\epsilon_{2345},-\epsilon_{1345},\epsilon_{1245},
-\epsilon_{1235},\epsilon_{1234})^T,
\quad o_i=\mathcal E^T\boldsymbol o_i,
\]
using $\epsilon_{ijkl}=\langle ij\rangle[jk]\langle kl\rangle[li]
-[ij]\langle jk\rangle[kl]\langle li\rangle$.
The complete coefficients of $o_0/c$ and the original $6\times12$ polynomial matrix are explicit polynomial inputs, not fitted low-energy coefficients.

For reconstruction independently of that polynomial table, the scalar
projection onto the five NMHV basis elements is generated as follows.
For each R index with occupied legs $i<j<k$, use
\[
C_{\ell,r}=\langle ij\rangle m_\ell[k]-\langle ik\rangle m_\ell[j]
+\langle jk\rangle m_\ell[i],\qquad
m_\ell=[\ell3]\eta_4+[34]\eta_\ell+[4\ell]\eta_3,
\quad \ell=1,2.
\]
The coefficients of $w^{1,2,3,4,0}$ in
$\prod_{r=1}^4(C_{1,r}+wC_{2,r})$, including the Grassmann permutation
sign from index ordering to leg ordering, give the five columns of the
six-scalar projection. Suppress the common denominator
$[34]^4\langle56\rangle^4$. Its left null row and the row with exchanged
spinors span the same two-dimensional space as $r,o$ on a dense open set.
The identity is verified by clearing denominators in generic spinor
coordinates. At exceptional row normalizations only the polynomial
identities are used; division by a vanishing normalization is not allowed.

The Gram matrix is
\begin{equation}
\mathsf G=\begin{pmatrix}
0&a&U-a-b&b+e-U-V&V-e-f_6\\
a&0&b&V-b-c&c+f_6-V-W\\
U-a-b&b&0&c&W-c-d\\
b+e-U-V&V-b-c&c&0&d\\
V-e-f_6&c+f_6-V-W&W-c-d&d&0
\end{pmatrix},\qquad \Delta_5=\det\mathsf G.
\label{eq:sm-physical-gram}
\end{equation}
It obeys $\mathsf G\mathcal E=0$. In
$R=\mathbb Q[a,b,c,d,e,f_6,U,V,W]$, the complete relations are
$\mathcal E^Tq=0$ on physical kinematics if and only if
$q\in\mathsf G R^5$. To see completeness, $\Delta_5$ is irreducible,
and $\operatorname{coker}\mathsf G$ is torsion-free over $R/(\Delta_5)$:
if $tq=\mathsf G b$ with $\Delta_5\nmid t$, multiplication by the
adjugate shows that $\operatorname{adj}(\mathsf G)q$ is divisible by
$\Delta_5$, and hence $q\in\mathsf G R^5$. The map from this cokernel to
the actual epsilon module is generically a rank-one isomorphism. Its kernel
would be torsion, so is zero. This accounts for Gram constraints even on the degenerate curve used by the functional; the dominance argument below identifies polynomial vanishing on physical kinematics with the scalar Gram ideal.

\paragraph{Coverage of the polynomial spinor verification.}
Let $\ell_i,\widetilde\ell_i$ denote the seven-parameter spinors
\[
\begin{split}
(\ell_1,\ldots,\ell_6)&=((1,0),(0,1),(1,1),(1,x),(1,y),(1,z)),\\
(\widetilde\ell_1,\ldots,\widetilde\ell_6)&=((-1-A-C,-1-B-D),(-1-yA-zC,-x-yB-zD),(1,0),(0,1),(A,B),(C,D)).
\end{split}
\]
Set $\lambda_i=\kappa\ell_i$ and $\widetilde\lambda_i=\widetilde\ell_i$. All six momenta are null and their sum vanishes identically. The invariant vector $v=(a,b,c,d,e,f_6,U,V,W)$ scales as $v=\kappa^2v|_{\kappa=1}$. At $(x,y,z,A,B,C,D,\kappa)=(2,3,5,1,2,3,4,1)$, it equals $(-7,28,-1,1,4,-5,14,8,-4)$, and
\[
\det\frac{\partial(a,b,c,d,e,f_6,U,V)}{\partial(x,y,z,A,B,C,D,\kappa)}=13608\ne0.
\]
The Gram matrix has rank four at this point. Exact factorization over $\mathbb Q$ gives $\Delta_5=2g$, where $g$ is one irreducible polynomial of degree five. Thus the invariant map has image of dimension eight inside the irreducible hypersurface $\Delta_5=0$, and is dominant. Equivalently, the kernel of the substitution map from $\mathbb Q[v]$ into the eight-parameter polynomial ring is exactly $(\Delta_5)$. Extension of scalars preserves this injection, so the same conclusion holds for real or complex coefficient identities.

The scalar projection matrix and its parity-exchanged matrix both scale by $\kappa^8$. The compact even Ward row scales by $\kappa^8$, while the full epsilon-bearing odd row scales by $\kappa^{12}$. Consequently, membership in the two Ward rows, verified polynomially at $\kappa=1$, also holds for arbitrary nonzero $\kappa$. The nonzero rank-five scalar-projection minor and the nonzero row-independence minor ensure that this statement is made on a dense open set. Clearing denominators then extends the resulting polynomial identities. The Jacobian witness establishes coverage; the separate exact polynomial calculations establish the Ward identities themselves.

\subsection{Polynomial five-point data and physical sewing}
\paragraph{Five-point polynomial dressing.}
At five points the MHV superamplitude is $\delta^{(8)}(Q)F_5$. Multiplication by the cyclic angle-bracket product gives the little-group-neutral dressing $G_5$. A physical two-particle pole in this sector sews a three-point anti-MHV amplitude to a four-point MHV amplitude, so its pole is on the corresponding angle-bracket branch; the cyclic product removes these allowed simple poles. To exclude a denominator hidden by a chosen component numerator, use all negative-helicity pairs: $A_{ij}=\langle ij\rangle^4F_5$. On any divisor where some $\langle ij\rangle\ne0$, locality of that component implies regularity of $G_5$ after the physical poles have been removed.

The remaining locus, where all angle brackets vanish, has codimension at least two in the five-point spinor momentum-conservation variety. A direct normality argument makes this sufficient. Write this variety as $X=\{(L,T)\in\mathrm{Mat}_{2\times5}\times\mathrm{Mat}_{5\times2}:LT=0\}$. The rank-two $L$ locus is an irreducible vector bundle of dimension $10+6=16$. Every lower-rank point lies in its closure: keep $T$ fixed and perturb the rows of $L$ inside its kernel, which has dimension at least three. The four equations have height four and form a complete intersection. The variety is smooth when either matrix has rank two; the locus where both have rank at most one has dimension at most eleven. It is therefore generically reduced, Cohen--Macaulay, and regular in codimension one, hence a normal integral variety. The locus $\operatorname{rank}L\le1$ has dimension at most fourteen. Thus regularity in codimension one extends across it, proving that the coefficientwise dressing is polynomial.

Little-group neutrality makes this a polynomial in the null momenta. Lorentz invariant theory expresses it in dot products and epsilon contractions. Momentum conservation leaves four independent five-point momenta, so every epsilon reduces to one epsilon and its square reduces to a Gram determinant. Hence
\[
G_5=V_5-\tfrac12\epsilon_5^{\rm sew}Q_5,
\]
with polynomial $V_5,Q_5$ in the five independent adjacent Mandelstams. The independence is also witnessed by the rank-five Jacobian in the accompanying certificate. Cyclicity preserves the five-point epsilon, so both polynomials are cyclic. This argument applies coefficientwise and does not impose reflection symmetry or set $Q_5$ to zero. The MHV conventions agree with Ref.~\cite{EHM:2026}, Eqs.~(2.16)--(2.19). The polynomial coverage argument is independent of finite-order counts.

\paragraph{Physical boundaries and the two chiral sewings.}
Use $(x_0,x_1,x_2,x_3,x_4)=(e,f_6,a,V,W)$ for the five-point order $[5,6,1,2,P]$, with $P=p_3+p_4$. On the holomorphic boundary $x_3=0$, direct spinor substitution gives
\[
\epsilon_5^{\rm sew}=J_3=ef_6-eW-af_6.
\]
The rank-four boundary Jacobian establishes that this is the full generic boundary. Its cyclic images give the five $J_i$ of the manuscript. Thus fixing the four-point factorization data gives $(V_5-J_iQ_5/2)|_{x_i=0}=B_i$, and differences at fixed four-point data obey the homogeneous version.

For the six-point cut $[34]=0$, set $\widetilde\lambda_3=\widetilde\lambda_4=\widetilde\lambda_P$ and $\lambda_P=\lambda_3+\lambda_4$. The three- and five-point scalar component prefactors, normalized to the SYM residue, have the Schouten reduction
\[
\frac{\langle34\rangle\langle52\rangle}{\langle2P\rangle\langle P5\rangle}
=\frac bV+\frac dW-1=\frac{bW+dV-VW}{VW}.
\]
The equality is checked after clearing denominators, so it does not assert division on their zero sets. The ten-parameter split family has a rank-seven Mandelstam Jacobian with a certified minor $-2206620$. Exact factorization shows that $\Delta_5|_{c=0}$ is irreducible, so this family is dominant on the physical cut. The mixed determinant in the manuscript obeys $\mathscr H=(bW+dV-VW)\epsilon_{1256}^{\rm sew}$ on this branch. Swapping spinors reverses epsilon, giving the opposite sign on the other branch while leaving $\mathscr H$ invariant.

The required relation between the two five-point sectors can be obtained from the stated six-scalar parity condition itself. Apply parity to the six-point residue, divide its generically nonzero SYM prefactor, and use the full five-point coverage of the split family. The anti-MHV dressing is then the spinor-exchanged MHV dressing on this cut. Thus the opposite epsilon signs give the same Mandelstam residue
\[
\mathcal R_{34}=\frac{(bW+dV-VW)V_5-\mathscr H Q_5/2}{VW}.
\]
This use of the already-assumed six-point parity does not add an independent five-point parity assumption. The relation concerns the universal dressings after their scalar component prefactors have been factored out.

\subsection{The full sewn source}
The five-point MHV dressing is $V_5-\epsilon^{\rm sew}_5Q_5/2$, where
$\epsilon^{\rm sew}=-\epsilon$ in the above Ward convention; its
anti-MHV partner has the opposite epsilon term. Evaluate $V_5,Q_5$ on
$(e,f_6,a,V,W)$ and put $r_f=f_6-V-W$. Define
\begin{equation}
\mathscr H=-\det\begin{pmatrix}
0&V&r_f&a\\ V&0&W&e-a-V\\
r_f&W&0&V-e-f_6\\ b&0&W-d&U-a-b
\end{pmatrix}.
\end{equation}
The physical two-particle residue of $Z_1$ is
$[(bW+dV-VW)V_5-\mathscr H Q_5/2]/(VW)$.
Both chiral factorization branches give this same parity-even expression.
An extension away from $c=0$ is
\begin{align}
G_6={}&-V_5-1+B(e,f_6)+B(a,f_6)
 +\frac bV[V_5-B(e,f_6)]\nonumber\\*
&+\frac dW[V_5-B(a,f_6)]-\frac{\mathscr H}{2VW}Q_5,
\label{eq:sm-sewing-G}\\
Z_1={}&\frac1c\left[1-\frac{k}{V}B(b,c)B(e,f_6)
-\frac zW B(c,d)B(f_6,a)+G_6\right]+C_6.
\end{align}
Here $C_6$ is an arbitrary local parity-even contact term. In
Eq.~\eqref{eq:sm-sewing-G}, the apparent $V,W$ denominators are removable.
Indeed, with $J_V=e(f_6-W)-af_6$, $J_W=a(f_6-V)-ef_6$,
\[
\mathscr H|_{V=0}=bWJ_V,\quad \mathscr H|_{W=0}=dVJ_W,
\quad
(V_5-J_VQ_5/2)|_{V=0}=B(e,f_6),\quad
(V_5-J_WQ_5/2)|_{W=0}=B(a,f_6).
\]
Thus the common numerator is divisible by $VW$ at every EFT order.
Changing the off-cut extension changes only $C_6$.

Let $P_U=B(a,b)B(d,e)$, $P_V=B(b,c)B(e,f_6)$,
$P_W=B(c,d)B(f_6,a)$, and define
\[
\Psi=xzh-yuk+(a+b-d-e)(a-c-d+f_6)(b+c-e-f_6),
\quad (\epsilon_U,\epsilon_V,\epsilon_W)
=(-\epsilon_{1346},\epsilon_{1245},\epsilon_{2356}).
\]
Writing $\operatorname{Tr}_6=\sum_i\rho^i$ and
$\operatorname{Tr}_6^-=\sum_i(-1)^i\rho^i$, the source components are
\begin{equation}
F_r=-\operatorname{Tr}_6[(r_0/c)G_6],\qquad
F_o=4\Psi\sum_J\epsilon_J(P_J-1)
-\operatorname{Tr}_6^-[(o_0/c)G_6].
\label{eq:sm-physical-source}
\end{equation}
The constant subtraction follows from the SYM Ward identity. Write
$F_o=\mathcal E^T\boldsymbol F_o$. The fixed matrix and source are
\begin{equation}
\mathbb M=\begin{pmatrix}
r_0&\cdots&r_5&\Delta_5&0_{1\times5}\\
\boldsymbol o_0&\cdots&\boldsymbol o_5&0_{5\times1}&\mathsf G
\end{pmatrix},\qquad
\mathbb F=\binom{F_r}{\boldsymbol F_o}.
\label{eq:sm-physical-M}
\end{equation}
Every physical local completion has $\mathbb F=\mathbb M C$ for polynomial
multipliers, coefficientwise in the EFT series. Allowing the six contact
multipliers to be independent only enlarges the space the annihilator must
kill. Conversely, cyclic averaging
$C_6=\frac16\sum_{i=0}^5\rho^{-i}C_i$ restores the common contact seed;
Gram columns vanish in the physical quotient. No inversion of a
nonplanar channel is part of this polynomial presentation.

\subsection{Why all five-point freedom is included}
Set $(x_0,x_1,x_2,x_3,x_4)=(e,f_6,a,V,W)$, let $\sigma$ rotate them,
and write $\operatorname{Tr}_5 h=\sum_{i=0}^4h(\sigma^i x)$.
At fixed four-point data the homogeneous five-point differences obey
$(\delta V_5-J_i\delta Q_5/2)|_{x_i=0}=0$, with
\begin{align*}
J_0&=-VW+Va-af_6,&J_1&=VW-Va-We,\\
J_2&=-VW+We-ef_6,&J_3&=-We-af_6+ef_6,&
J_4&=-Va+af_6-ef_6.
\end{align*}
Comparing adjacent boundaries forces
$\delta Q_5\in\bigcap_i(x_i,x_{i+1})$. This monomial ideal has exactly
the five cyclic generators $f_6VW,aeW,ef_6V,af_6W,aeV$.
The seed
\[
q_*=f_6VW,\qquad
v_*=-\tfrac12 f_6VW(VW-Va-We+af_6+ef_6)
\]
lifts its first generator through all five boundaries. Cyclic averaging,
and divisibility of a function vanishing on every coordinate hyperplane
by $\Pi=ef_6aVW$, prove the complete parameterization
\begin{equation}
\delta Q_5=\operatorname{Tr}_5(q_*T),\qquad
\delta V_5=\operatorname{Tr}_5(v_*T)+\Pi P.
\label{eq:sm-five-complete}
\end{equation}
$T$ is arbitrary and $P$ is cyclic; this parameterization may be redundant.
At four-point degree $N$, $\deg Q_5=N$, $\deg V_5=N+2$, and
$\deg T=\deg P=N-3$. The canonical particular solution excludes contact
ideal monomials in $Q_5$ and $\Pi$-divisible monomials in $V_5$; the
remaining terms are fixed by finite inclusion--exclusion on the boundaries.
In particular,
\[
Q_5(e,f_6,a,0,0)=\frac{a f(a,f_6)-e f(e,f_6)}{a-e},\qquad
V_5(e,f_6,a,0,0)=\tfrac12[B(a,f_6)+B(e,f_6)].
\]
All divisions are removable divided differences. This construction retains
five-point parity-odd data $Q_5$; the assumed six-scalar parity condition
does not set $Q_5$ to zero.

For a four-point term of Mandelstam degree $N$, the dressing correction has degree $N+2$. Since epsilon has degree two, $\deg Q_5=N$ and $\deg V_5=N+2$. The generators have degrees three and five, respectively, so $\deg T=\deg P=N-3$; no independent five-point contact freedom occurs for $N<3$. The sewn $G_6$ has degree $N+2$ and the six-point contact $C_6$ has degree $N+1$. Both coefficient rows of the $6\times12$ matrix have contact-column degree four, so the source coefficient vector has degree $N+5$. Its scalar Gram multiplier has degree $N$, and its five Gram--epsilon multipliers have degree $N+4$. These are homogeneous degree statements, not extrapolations from a finite basis count.

\paragraph{Local extension and six-point completeness.}
The determinant identities $\mathscr H|_{V=0}=bWJ_V$ and $\mathscr H|_{W=0}=dVJ_W$, together with the five-point boundary equations, imply that the common numerator of $G_6$ vanishes on both coordinate hyperplanes. The five-point variables are independent, so it is divisible by $VW$ in the polynomial ring at every degree. No nonplanar denominator is introduced.

For the scalar states $(12,13,14,23,34,24)$, R-charge counting permits only the two-particle channel $s_{34}$ and the three-particle channels $s_{234},s_{345}$. The latter residues are products of four-point SYM components multiplied by the corresponding $B$ factors. Subtracting these and the just-verified two-particle residue leaves a local parity-even scalar contact term. Such a term is a polynomial in the Mandelstams modulo the scalar Gram relation, since even products of epsilon contractions reduce to Gram determinants. All its cyclic images are included in the first six matrix columns; treating their coefficients independently enlarges the allowed space and therefore remains valid for necessity.

More explicitly, if two polynomial off-cut numerators have the same physical $c=0$ residue, their difference belongs to $(c,\Delta_5)$: the cut is an irreducible reduced hypersurface and the split family is dominant on it. In the physical ring $R/(\Delta_5)$, their difference is consequently a multiple of $c$, and division by the propagator changes only a local contact term. This justifies the off-cut extension on the Gram variety, including the degenerate curve. Combined with the previously verified complete Gram--epsilon module, it gives the required polynomial implication $\mathbb F\in\operatorname{im}_R\mathbb M$ for every physical local completion. It does not assert sufficiency for an arbitrary collection of scalar data to extend to all superamplitude components.

\section{All-order boundary Ward functional and kernel necessity}
\label{sm:boundary}
The curve parameter $t$ in this section is the parameter $\tau$ of the
Letter; it is held fixed when taking kinematic derivatives. All nonzero
kinematic evaluations are performed on homogeneous coefficient polynomials,
before summing an EFT series.
\begin{theorem}
Consider a four-dimensional, tree-level, local, planar single-trace
$\mathcal N=4$ SYM EFT with unbroken $SU(4)$, nonzero SYM three-point coupling,
the specified massless multiplet in factorization channels, and the scalar
parity assumption used in the six-component Ward presentation. Exclude
spurious or nonplanar poles and simultaneous poles in incompatible channels.
Normalize the four-point scalar function as
\[
A(s,u)=-\frac1{su}+f(s,u),\qquad f(s,u)=f(u,s),\qquad B(s,u)=1-su f(s,u).
\]
Then, as formal series,
\begin{equation}
zH'(z)+3H(z)-z^2H(z)^2=6c_{00},\qquad
H(z)=f(z,0)+f(-z,0),\quad c_{00}=f(0,0).
\label{sm-boundary-eq:riccati}
\end{equation}
Together with the physical cocycle proved in Sec.~\ref{sm:cocycle}, this implies
\begin{equation}
q_A''=-6c_{00}q_A,\qquad q_A(z)=\frac{z}{B(z,-z)},\qquad
q_A(0)=0,\quad q_A'(0)=1.
\label{sm-boundary-eq:qode}
\end{equation}
The scalar parity condition is an additional assumption; it is not implied
by maximal supersymmetry alone.
\end{theorem}

\paragraph{Ward framework.}
Use the Mandelstam coordinates
$v=(a,b,c,d,e,f_6,U,V,W)$ and the original polynomial matrix
\[
\M=\begin{pmatrix}
r_0&\cdots&r_5&\Delta_5&0_{1\times5}\\
\boldsymbol o_0&\cdots&\boldsymbol o_5&0_{5\times1}&\mathsf G
\end{pmatrix}.
\]
The six contact columns are the two physical scalar Ward rows, with the odd
row expanded in the five-epsilon basis. The remaining columns are the scalar
Gram relation and the complete Gram--epsilon relations. In the presentation of Sec.~\ref{sm:physical-presentation}, physical locality requires $\F\in\operatorname{im}\M$.
All calculations below use this original matrix, not just its leading terms.

\paragraph{A second-order curve functional.}
Set
\[
\Gamma(t)=(t,0,-t,1,0,-1,2t,-2t,0),\qquad \mu=(4,0,1,-1,0,0),
\]
\begin{align*}
r(t)=\bigg(&\frac{-4t^3+20t^2-14t+2}{t^2(t-1)^2},
-\frac{2t^2-2t+1}{2t^2(t-1)^2},-\frac2t,0,\\
&-\frac{4t^3-8t^2+6t-1}{2t^2(t-1)^2},
-\frac{2t-1}{t(t-1)}\bigg),
\end{align*}
\[
\eta(t)=\left(-\frac{4(2t-1)}t,\frac{2t-1}{t(t-1)},
-\frac{2(2t-1)}t,0,\frac{2t-1}{t(t-1)},0\right).
\]
Define, with $t$ fixed during kinematic differentiation,
\begin{equation}
J_t[\F]=\left[
r(t)\F+\eta(t)\partial_e\F+
\mu\left(\partial_c+\frac{2t}{t-1}\partial_U
+(2t-1)\partial_e(\partial_c-2\partial_U)\right)\F
\right]_{\Gamma(t)}.
\label{sm-boundary-eq:jt}
\end{equation}

\subsection{All-order annihilation}
For every even $N\ge2$, define
\begin{equation}
 R_t=J_t/K(t),\qquad \mathcal L_N=4[t^{N-2}]R_t,\qquad K(t)=-32(t-1)(2t-1)^2(3t^2-3t+1).
 \label{sm-boundary-eq:ln}
\end{equation}
\paragraph{A complete product-rule certificate.}
Work at fixed homogeneous EFT degree in $R=\mathbb Q[v]$, retaining the original $6\times12$ matrix $\mathbb M$. Write $h=2t/(t-1)$, $k=2t-1$, $d_t=2t^2(t-1)^2$, and use the stated $r(t),\eta(t),\mu$ and $\Gamma(t)$ in the definition of $J_t$. Kinematic derivatives are taken before restriction, with $t$ held fixed. The complete contact certificate can be written as four row-vector identities:
\begin{align*}
 [\mu\mathbb M]_{\Gamma(t)}&=0,&[\mu\partial_e\mathbb M]_{\Gamma(t)}&=0,\\
 [d_t\eta\mathbb M+d_t k\mu(\partial_c-2\partial_U)\mathbb M]_{\Gamma(t)}&=0,&d_tJ_t[\mathbb M]&=0.
\end{align*}
All 48 scalar entries vanish as polynomials in $t$. These identities include the six scalar contact columns, the scalar Gram column, and all five Gram--epsilon columns.

To see that this finite certificate covers arbitrary local multipliers at every degree, let $m$ be any column and $C\in R$. The ordinary Leibniz rule gives
\begin{align*}
 J_t[mC]={}&J_t[m]\,C+[\eta m+k\mu(m_c-2m_U)]C_e+(\mu m+k\mu m_e)C_c\\
 &+(h\mu m-2k\mu m_e)C_U+k\mu m\,C_{ce}-2k\mu m\,C_{eU},
\end{align*}
where all quantities on the right are evaluated at $\Gamma(t)$. The four row identities set every coefficient to zero. No derivative in any of the remaining six kinematic directions can occur. Equivalently, the original 180 identities for Taylor monomials of degree at most two in $(a,c,e,U)$ are a redundant complete certificate; including all nine variables gives 660 identities, with the additional unused directions vanishing structurally. This is a finite verification of a differential-operator identity, with no upper bound on the degree of $C$.

Consequently $J_t$ depends only on the class in $R^6/\operatorname{im}\mathbb M$. In particular, changing an off-shell polynomial representative by a scalar Gram or Gram--epsilon relation has zero response, even though the separate kinematic derivatives need not be tangent to the Gram hypersurface. No division by a Gram minor on the curve is used.

\paragraph{The complete five-point response.}
Let $x=(e,f_6,a,V,W)$, $\Pi=ef_6aVW$, $q_*=f_6VW$, and $v_*=-q_*(VW-Va-We+af_6+ef_6)/2$. The previously established complete homogeneous freedom is $\delta Q_5=\operatorname{Tr}_5(q_*T)$ and $\delta V_5=\operatorname{Tr}_5(v_*T)+\Pi P$, with $T$ arbitrary and $P$ cyclic. For each of the five cyclic rotations $q_k,v_k$, the sewn coefficient is the polynomial
\[
 G_k=\frac{(bW+dV-VW)v_k-\mathscr H q_k/2}{VW},
\]
where $\mathscr H$ is the mixed determinant in the physical sewing formula; the numerator is exactly divisible by $VW$ in $R$. The $P$ coefficient is $ef_6a(bW+dV-VW)$. These exact divisions are performed before any substitution with $W=0$.

Applying the full ordinary Leibniz rule to these six generators in all six cyclic sewings, and keeping the values and derivatives of $T$ and $P$ at their respective five-point arguments, gives
\[
 \delta J_t=K(t)\bigl[T(0,-1,2t,-2t,1)-T(-2t,1,0,-1,2t)\bigr],\qquad K(t)=-32(t-1)(2t-1)^2(3t^2-3t+1).
\]
Every first- and second-derivative coefficient of $T$ vanishes. Every coefficient of $P$ and its derivatives vanishes without requiring the cyclicity of $P$. The two displayed undifferentiated coefficients are $K$ and $-K$, and no other function value survives. The calculation uses the complete five-point generators and the original matrix; it assumes no relation on the four-point input.

\paragraph{Coefficient extraction and differential order.}
At four-point degree $N$, the five-point dressing degrees are $\deg Q_5=N$ and $\deg V_5=N+2$, so the complete homogeneous freedom can be represented with $\deg T=\deg P=N-3$. Since each displayed argument is affine in $t$, the response $\delta J_t/K(t)$ is a polynomial of degree at most $N-3$. It follows that $4[t^{N-2}]\delta J_t/K(t)=0$ for every even $N\ge4$; there is no homogeneous freedom for $N=2$. This is a degree argument for arbitrary multipliers, not an extrapolation from low EFT orders.

The coefficient $K(0)=32$ is a unit in $\mathbb Q[[t]]$. Each coefficient $w_\alpha(t)$ of a derivative $\partial^\alpha$ in $J_t/K(t)$ satisfies $\operatorname{ord}_{t=0}w_\alpha\ge |\alpha|-2$. Writing $\Gamma(t)=\Gamma(0)+t\xi$, where $\xi\cdot\partial=\partial_a-\partial_c+2\partial_U-2\partial_V$, yields, for each polynomial $F$,
\[
 4[t^{N-2}]\frac{J_t[F]}{K(t)}
 =4\sum_{r,\alpha}\sum_{j=0}^{N-|\alpha|}[t^{N-2-j}]w_{r,\alpha}(t)\,
 \frac{(\xi\cdot\partial)^j\partial^\alpha F_r(\Gamma(0))}{j!}.
\]
Thus the extracted functional has differential order at most $N$ and is well defined using Laurent coefficients with only finitely many negative powers. Evaluation at $\Gamma(0)$ is made only after fixing the homogeneous EFT degree, when the source is polynomial; no convergence of the full EFT series at nonzero kinematics is assumed. The contact identity and the five-point degree bound therefore prove the required annihilation at every allowed degree.

\subsection{The source for arbitrary four-point input}
\paragraph{A universal certificate for the linear source.}
Introduce independent formal parameters $p,q$ and the auxiliary generating input
\[
 f^{p,q}(s,u)=\frac{1}{(1-ps)(1-qu)}+\frac{1}{(1-qs)(1-pu)}.
\]
Its coefficient of $p^iq^j$ is $s^iu^j+s^ju^i$. These coefficients span all symmetric polynomials over a field of characteristic zero. The canonical five-point lift, the linear six-point source, and the curve functional are linear in the four-point input and commute with extraction of these auxiliary coefficients. Thus one exact rational identity for this input proves the linear source identity at every homogeneous EFT degree. The auxiliary rational expression is used through its formal expansion in $p,q$; each coefficient is a polynomial local input.

For $x=(e,f_6,a,V,W)$ put $D_x=\prod_{j=0}^4(1-px_j)(1-qx_j)$. All canonical lift terms have denominator dividing $D_x$, whose constant term in $p,q$ is one. In particular the divided-difference seed becomes
\[
 \frac{a f^{p,q}(a,f_6)-e f^{p,q}(e,f_6)}{a-e}
 =\frac{1}{(1-pa)(1-pe)(1-qf_6)}+\frac{1}{(1-qa)(1-qe)(1-pf_6)}.
\]
This identity removes the apparent $a=e$ singularity before differentiation. Write the canonical $Q_5$ and linear correction to $V_5$ as $N_Q/D_x$ and $N_V/D_x$. Their numerators are obtained by inclusion--exclusion on the five codimension-two faces for $Q_5$ and the five coordinate faces for $V_5$. Equivalently, coordinate restriction on a numerator over $D_x$ acts by
\[
 E_S(N)=\left.N\right|_{x_j=0,\,j\in S}\prod_{j\in S}(1-px_j)(1-qx_j).
\]
The independent certificate checks all ten face equations and the cyclicity of both numerators as exact polynomial identities. Inclusion--exclusion retains precisely the monomials outside the respective contact ideals, so this is the canonical lift used in the source formula.

Let $N_{ef}$ and $N_{af}$ denote the numerators of $-ef_6f^{p,q}(e,f_6)$ and $-af_6f^{p,q}(a,f_6)$ over $D_x$. With $\mathscr H$ the physical mixed determinant, the numerator
\[
 (bW+dV-VW)N_V+VW(N_{ef}+N_{af})-bWN_{ef}-dVN_{af}-\tfrac12\mathscr H N_Q
\]
is exactly divisible by $VW$ in $\mathbb Q[a,b,c,d,e,f_6,U,V,W,p,q]$. The quotient is formed globally before restricting to $\Gamma(t)$, including its locus $W=0$. The remaining denominator $D_x$ is a unit as a formal series in $p,q$. Therefore the source has well-defined polynomial coefficients at every EFT degree, and the kinematic derivatives and curve restriction are legitimate. The certificate differentiates the six cyclic sewings and adds the linear contributions from all three quadratic-factorization channels. The resulting rational identity in $\mathbb Q(t,p,q)$ is
\[
 R_t[\mathbb F_{\rm lin}(f)]=\frac{2f_+(2t,-1)-f_+(t,-1)-f_+(2t,-2t)}{2t(1-2t)}-\frac{f_+(0,1)-c_{00}}{2t},\qquad f_+(s,u)=\frac{f(s,u)+f(-s,-u)}2.
\]
Coefficient extraction in $p,q$ proves the formula for arbitrary symmetric polynomial $f$, and hence coefficientwise for a formal EFT input. This proof does not impose a cocycle or a special four-point solution.

\paragraph{All quadratic channels.}
The quadratic source has three channel multipliers with argument pairs $(a,b;d,e)$, $(b,c;e,f_6)$, and $(c,d;f_6,a)$. The first two carry a factor $b$, while $J_t$ differentiates only in $c,e,U$ and the curve has $b=0$. Their entire responses, including all function-derivative terms, vanish. For the third channel the multiplier $f(c,d)f(f_6,a)$ depends, among the differentiated coordinates, only on $c$. The exact product-rule coefficient of this derivative vanishes, while its value coefficient after dividing by $K$ is $-1/[4(1-2t)]$. Consequently
\[
 R_t[\mathbb F_{\rm quad}(f)]=-\frac{f(-t,1)f(-1,t)}{4(1-2t)}.
\]
The calculation keeps both even and odd homogeneous pieces of $f$ and includes every allowed quadratic cut. The independent certificate verifies the surviving coefficient as well as every required zero. The quadratic certificate asserts its exact value explicitly.

\paragraph{All-degree coefficient extraction.}
Write $f=\sum_{n\ge0}f_n$ with $f_n$ homogeneous of degree $n$. At source degree $N+5$, the linear contribution uses $f_N$ and the quadratic contribution uses all ordered pairs $f_kf_l$ with $k+l=N-2$. Fix even $N\ge2$, put $g(t)=f_N(t,-1)=\sum_{j=0}^Ng_jt^j$, and set $d_N=f_N(1,-1)$, $a_N=f_N(1,0)$. In the linear source, the term containing $f_N(2t,-2t)$ starts at power $t^{N-1}$, and the term $a_N/(2t)$ has only power $t^{-1}$. Neither contributes to $[t^{N-2}]$. Thus
\begin{align*}
 [t^{N-2}]R_t[\mathbb F_{\rm lin}(f_N)]
 &=\frac12\sum_{j=0}^{N-1}(2^{j+1}-1)2^{N-1-j}g_j\\
 &=2^{N-1}d_N-2^{N-2}f_N(1/2,-1)-(2^{N-1}-1/4)a_N.
\end{align*}
For the second equality use $g_N=a_N$ and $g(1)=d_N$. Homogeneity, symmetry, and the evenness of $N$ identify four times this expression with the degree-$N$ part of
\[
 2f(2z,-2z)-\tfrac12[f(2z,-z)+f(z,-2z)]-H(2z)+\tfrac12H(z),\qquad H(z)=f(z,0)+f(-z,0).
\]

For the quadratic term let $m=N-2$ and $G_m(t)=\sum_{k+l=m}f_k(-t,1)f_l(-1,t)$. Its degree in $t$ is at most $m$. The finite geometric-sum identity gives
\[
 [t^m]\frac{G_m(t)}{1-2t}=2^mG_m(1/2).
\]
Symmetry and homogeneity imply $f_k(-1/2,1)=2^{-k}f_k(2,-1)$ and $f_l(-1,1/2)=2^{-l}f_l(1,-2)$. It follows that four times the extracted quadratic response is
\[
 -\sum_{k+l=N-2}f_k(2,-1)f_l(1,-2)=[z^N]\{-z^2f(2z,-z)f(z,-2z)\}.
\]
No even-degree restriction is imposed on the individual $k,l$ terms in this sum.

Combining both pieces proves $\mathcal L_N[\mathbb F_{{\rm can},N}]=[z^N]\mathcal P(z)$ for every even $N\ge2$, with
\[
 \mathcal P(z)=2f(2z,-2z)-\tfrac12[f(2z,-z)+f(z,-2z)]-H(2z)+\tfrac12H(z)-z^2f(2z,-z)f(z,-2z).
\]
Under $z\mapsto-z$, symmetry exchanges the two factors $f(2z,-z)$ and $f(z,-2z)$ and preserves $f(2z,-2z)$ and $H(z)$. Hence $\mathcal P$ is even. Its constant coefficient is $2c_{00}-c_{00}-2c_{00}+c_{00}=0$. Physical locality and the previously verified all-degree contact and five-point annihilation give $\mathcal L_N[\mathbb F_{{\rm can},N}]=0$. Therefore $\mathcal P(z)=0$ as a formal series, before invoking the cocycle.

\paragraph{Legality of the formal operations.}
All operations are performed at fixed homogeneous EFT degree. The factors $1-2t$ and $K(t)$ are units at $t=0$; negative powers of $t$ are interpreted in the Laurent-series ring. The order bound above makes every coefficient a finite differential functional. No infinite EFT series is evaluated at the nonzero coordinates of $\Gamma(t)$.

\subsection{Kernel equation and monodromy}
\paragraph{From the source identity to a doubling equation.}
The cocycle gives $B(2z,-z)B(z,-2z)=D(2z)$. Writing $F_1=f(2z,-z)$ and $F_2=f(z,-2z)$, this is $(1+2z^2F_1)(1+2z^2F_2)=D(2z)$. The independently verified source constraint therefore simplifies to
\[
 \mathcal P(z)=\frac{D(2z)-1}{4z^2}-H(2z)+\frac12H(z)=0.
\]
The quotient by $z^2$ is regular because $D-1\in z^2\mathbb K[[z]]$. Equivalently, $D(z)-1=z^2[H(z)-H(z/2)/2]$. Using
\[
 H(z)=\frac{1-zq'(z)/q(z)}{z^2}
\]
at $z$ and $2z$ converts this identity to
\[
 [1+q'(2z)]q(z)=q(2z)q'(z).
\]
Here $q'(2z)$ is the derivative of $q$ in its argument evaluated at $2z$. Divisions by $q=z\,u(z)$, with $u(0)=1$, are made in the formal Laurent ring; clearing these factors and cancelling powers of $z$ is injective. No solution with the prescribed linear coefficient is lost or added.

\paragraph{All-order uniqueness and the differential equation.}
Write $q(z)=z+\sum_{n\ge1}b_{2n+1}z^{2n+1}$. In the coefficient of $z^{2n+1}$ in the doubling equation, the new coefficient $b_{2n+1}$ has multiplier
\[
 M_n=(2n-1)4^n-4n.
\]
This follows by varying $q$ by $\varepsilon z^m$: its contribution at order $z^m$ is $[2+m2^{m-1}-(2m+2^m)]\varepsilon$, with $m=2n+1$. All other terms at that order involve earlier coefficients only. The multiplier vanishes for $n=1$, leaving the cubic datum free, and for every $n\ge2$ it obeys $M_n\ge28n-16>0$. Therefore the cubic datum $b_3=-c_{00}$ determines at most one normalized odd solution to all orders.

Set $\kappa=6c_{00}$ and define the normalized formal ODE solution without introducing a square root:
\[
 S_\kappa(z)=\sum_{n\ge0}\frac{(-\kappa)^nz^{2n+1}}{(2n+1)!},\qquad C_\kappa=S_\kappa'.
\]
It satisfies $S_\kappa''=-\kappa S_\kappa$, the initial conditions $S_\kappa(0)=0$, $C_\kappa(0)=1$, and the first integral $C_\kappa^2+\kappa S_\kappa^2=1$. Formal ODE uniqueness gives the addition laws
\[
 S_\kappa(x+y)=S_\kappa(x)C_\kappa(y)+C_\kappa(x)S_\kappa(y),\qquad C_\kappa(x+y)=C_\kappa(x)C_\kappa(y)-\kappa S_\kappa(x)S_\kappa(y).
\]
These imply $S_\kappa(2z)=2S_\kappa C_\kappa$ and $1+C_\kappa(2z)=2C_\kappa^2$, so $S_\kappa$ obeys the same doubling equation. Its cubic coefficient is $-\kappa/6=-c_{00}$. The all-order uniqueness just proved gives $q=S_\kappa$ and therefore $q''=-6c_{00}q$.

For real $c_{00}$ the three normalized branches are
\[
 q(z)=\begin{cases}
 z,&c_{00}=0,\\
 \sin(\sqrt{6c_{00}}\,z)/\sqrt{6c_{00}},&c_{00}>0,\\
 \sinh(\sqrt{-6c_{00}}\,z)/\sqrt{-6c_{00}},&c_{00}<0.
 \end{cases}
\]
All are represented by the same power series in $\kappa$; the derivation never divides by $c_{00}$ or its square root. The identity $q'/q=1/z-zH$ yields $q''/q=-(zH'+3H-z^2H^2)$, proving the stated Riccati equation. In particular, if $a_k=[z^k]f(z,0)$, then $a_{2n}=2\sum_{i+j=n-1}a_{2i}a_{2j}/(2n+3)$ for every $n\ge1$.

\paragraph{Four-point monodromy and its scope.}
For $s+t+u=0$, cochain reconstruction gives
\[
 (A(s,t),A(s,u),A(t,u))=\Omega\,(q(u),q(t),q(s)),\qquad\Omega=-\frac{\gamma(s)\gamma(t)\gamma(u)}{stu}.
\]
Choose $\lambda^2=\kappa$ over a complex extension. The initial-value problem implies $e^{i\lambda z}=C_\kappa(z)+i\lambda S_\kappa(z)$, including $\lambda=0$. The addition law then gives
\[
 q(u)+e^{i\lambda s}q(t)+e^{-i\lambda t}q(s)=0,
\]
because the imaginary cross terms cancel and the real part is $-q(s+t)+q(s)C_\kappa(t)+C_\kappa(s)q(t)$. Multiplying by $\Omega$ establishes the manuscript's monodromy relation. For $c_{00}=0$ it reduces regularly to the sum of the three orderings; for negative $c_{00}$ it has the corresponding real exponential factors. Both signs of $\lambda$ give valid identities and the same normalized kernel.

The equations are identities of formal series after clearing the massless poles. The kernel itself is represented by an entire series for each fixed finite $c_{00}$, but this does not imply convergence of the full four-point amplitude: the surviving cochain freedom is
\[
 \gamma(z)=\sqrt{\frac{z}{q(z)}}\,e^{\phi(z)},\qquad\phi(z)\in z^3\mathbb K[[z^2]],
\]
where the square root is the unique formal root with constant term one. The arbitrary odd series $\phi$ can be divergent. Analytic uniqueness of the amplitude requires the later positive-moment argument or additional convergence information; no global continuation follows merely from the physical necessity theorem. These statements characterize the necessary four-point constraints and do not assert a higher-point completion for every surviving $\phi$.

\section{Positive axis measure and conditional Veneziano uniqueness}
\label{sm:positive}
\begin{theorem}
Let $A(s,t)=-1/(st)+f(s,t)$ be a symmetric formal four-point germ satisfying
the physical cocycle and Riccati conditions proved in Sec.~\ref{sm:boundary}. Suppose its axis coefficients admit one common finite nonnegative
Borel measure on the nonnegative real axis:
\begin{equation}
f(s,0)=\sum_{k\ge0}a_k s^k,\qquad
a_k=\int_0^\infty x^k\,d\mu(x),\quad k\ge0.
\label{sm-positive-eq:mom}
\end{equation}
All these moments are assumed finite. Then either $\mu=0$ and $A=-1/(st)$,
or $\alpha=\sqrt{6a_0}/\pi>0$ and
\begin{equation}
A(s,t)=-\alpha^2
\frac{\Gamma(-\alpha s)\Gamma(-\alpha t)}{\Gamma(1-\alpha(s+t))}.
\label{sm-positive-eq:ven}
\end{equation}
Equality holds as a formal series and as its unique analytic realization: for $\alpha>0$, the pole-subtracted $f$ is holomorphic on $|s|+|t|<1/\alpha$. No prior compact support of $\mu$ or convergence of the initial formal series is assumed.
\end{theorem}

\begin{proof}
If $a_0=0$, positivity gives $\mu=0$, hence every axis coefficient vanishes. The cocycle reconstruction $(\log\gamma)'(z)=zf(z,0)$ then gives $\gamma=1$ in the fixed linear gauge and $A=-1/(st)$. In particular, zero axis data cannot hide additional mixed coefficients once the cocycle is imposed. Suppose henceforth $a_0>0$ and set $\alpha=\sqrt{6a_0}/\pi>0$.

\paragraph{Even moments and support.}
The normalized kernel is $q(z)=\sin(\pi\alpha z)/(\pi\alpha)$. Therefore
\[
 H(z)=\frac{1-\pi\alpha z\cot(\pi\alpha z)}{z^2}
 =2\alpha^2\sum_{m\ge1}\frac{1}{m^2-\alpha^2z^2}
 =2\sum_{n\ge0}\alpha^{2n+2}\zeta(2n+2)z^{2n}.
\]
The first expansion uses the cotangent partial fraction identity, \href{https://dlmf.nist.gov/4.22.E3}{DLMF 4.22.3}. For $|\alpha z|\le r<1$, the summands of $H$ are bounded by $2\alpha^2/[m^2(1-r^2)]$; the resulting geometric double series is absolutely and locally uniformly convergent. Thus coefficient extraction is justified independently of any convergence assumption on the starting formal amplitude. Comparing with $H(z)=2\sum a_{2n}z^{2n}$ gives $a_{2n}=\alpha^{2n+2}\zeta(2n+2)$ for every $n\ge0$.

For every $L>\alpha$ and every $n\ge0$, positivity gives
\[
 \mu([L,\infty))\le L^{-2n}a_{2n}\le\alpha^2\zeta(2)(\alpha/L)^{2n}.
\]
Letting $n\to\infty$ shows that this tail has zero mass. A countable union over $L\downarrow\alpha$ yields $\mu((\alpha,\infty))=0$, so $\operatorname{supp}\mu\subseteq[0,\alpha]$. This is a consequence of the entire even-moment sequence, not an extra support hypothesis or a finite-order test.

\paragraph{Measure determinacy, including the endpoint at zero.}
The finite positive measure
\[
 \mu_* =\sum_{m\ge1}\frac{\alpha^2}{m^2}\delta_{\alpha/m}
\]
has exactly the same even moments as $\mu$, including total mass $a_0=\alpha^2\zeta(2)$. Push both measures forward by $x\mapsto x^2$ to $[0,\alpha^2]$. Their integrals agree on every polynomial. Polynomial uniform approximation of continuous functions, together with the finite common mass, extends this equality to every continuous function. Uniqueness of finite Borel measures on a compact interval then makes the pushforwards equal. The square map is a homeomorphism from $[0,\alpha]$ to $[0,\alpha^2]$, so its inverse square root pulls this equality back to $\mu=\mu_*$. Including the zeroth moment excludes an undetected atom at $x=0$. The restriction to the nonnegative axis makes the square map injective.

Consequently, for every $k\ge0$,
\[
 a_k=\alpha^{k+2}\zeta(k+2).
\]
Thus the single common measure determines every odd axis coefficient from the even ones. The physical cocycle subsequently determines every mixed four-point coefficient; no family of angular-derivative moment assumptions is needed.

\paragraph{A constructive all-order bound.}
For completeness, put $r_M=\binom{2M}{M}/4^M$ and $b_j=\binom{2j}{j}/[4^j(2j-1)]$. The exact identities $b_j=r_{j-1}-r_j$ and $r_j/r_{j-1}=(2j-1)/(2j)$ show that $\sum_{j>M}b_j=r_M$. An elementary induction gives $r_M\le1/\sqrt{M+1}$: its induction step follows from
\[
 4(M+1)^3-(2M+1)^2(M+2)=3M+2>0.
\]
The positive binomial expansion $\sqrt y=1-\sum_{j\ge1}b_j(1-y)^j$ extends uniformly to $[0,1]$, and hence
\[
 L_M(y)=1-\sum_{j=1}^M b_j(1-y)^j-r_M(1-y)^{M+1}\le\sqrt y\le U_M(y)=1-\sum_{j=1}^M b_j(1-y)^j.
\]
If $\nu$ is the pushforward of $\mu$ by $y=(x/\alpha)^2$, integrating $y^pL_M$ and $y^pU_M$ bounds $a_{2p+1}/\alpha^{2p+1}$ using only even moments. The resulting interval for the unnormalized moment $a_{2p+1}$ has width at most $\alpha^{2p+1}a_0/\sqrt{M+1}$ and therefore shrinks to zero for every fixed $p$. This proves constructive odd-moment uniqueness at arbitrary order. Finite decimal intervals check these bounds; the argument above proves determinacy at every order.

\paragraph{The formal series is convergent.}
The support conclusion immediately gives $|a_k|=a_k\le a_0\alpha^k$. Thus the axis series $F(z)=\sum a_kz^k$ converges for $|z|<1/\alpha$, and
\[
 F(z)=\int_{[0,\alpha]}\frac{d\mu(x)}{1-zx},\qquad |F(z)|\le\frac{a_0}{1-\alpha|z|}.
\]
The integral and sum commute by uniform convergence of the geometric series on the support. The normalized cochain logarithm $h(z)=\int_0^z wF(w)\,dw$ is holomorphic in the same disk, with $|h(z)|\le a_0|z|^2/[2(1-\alpha|z|)]$. Therefore
\[
 B(s,t)=\exp\{h(s)+h(t)-h(s+t)\}
\]
is holomorphic on $|s|+|t|<1/\alpha$. Its numerator $1-B(s,t)$ vanishes identically on each coordinate axis. Successive holomorphic division by $s$ and $t$ makes $f=(1-B)/(st)$ holomorphic there as well. Hence the formal solution defines a unique convergent local analytic germ, even when the input was specified only as a formal EFT series. Identifying an independently given physical function with this realization still uses analyticity of that function after subtracting the massless poles.

\paragraph{Veneziano normalization.}
The determined moments give
\[
 h(z)=\sum_{r\ge2}\frac{\zeta(r)}r(\alpha z)^r
 =\sum_{m\ge1}\left[-\log(1-\alpha z/m)-\alpha z/m\right].
\]
Both sums converge absolutely for $|\alpha z|<1$, and the omitted linear term fixes the cochain gauge. The reciprocal Gamma product, \href{https://dlmf.nist.gov/5.8.E2}{DLMF 5.8.2}, identifies
\[
 \gamma(z)=e^{h(z)}=e^{-\gamma_E\alpha z}\Gamma(1-\alpha z).
\]
The linear exponent cancels in $B=\gamma(s)\gamma(t)/\gamma(s+t)$. Using $\Gamma(1-\alpha s)=(-\alpha s)\Gamma(-\alpha s)$ and the corresponding identity for $t$ yields
\[
 A(s,t)=-\frac{\Gamma(1-\alpha s)\Gamma(1-\alpha t)}{st\,\Gamma(1-\alpha(s+t))}
 =-\alpha^2\frac{\Gamma(-\alpha s)\Gamma(-\alpha t)}{\Gamma(1-\alpha(s+t))}.
\]
This identifies the full formal series and its convergent local germ. The Gamma expression supplies a mathematical meromorphic continuation of that germ. Equality with a physical amplitude beyond the low-energy neighborhood holds on any common connected meromorphic domain where that physical continuation exists; the coefficient argument does not construct such a physical domain or a complete higher-point S matrix.

\end{proof}
\paragraph{Role of the positive-measure premise.}
One common measure represents the entire sequence, including $a_0$; separate coefficientwise measures do not suffice. An unsubtracted positive spectral representation with justified coefficient extraction implies this premise. It is additional to local massless factorization and scalar parity, and does not permit an arbitrary subtraction polynomial. Recovery of odd moments from even ones is a compact determinacy argument also used at axis order in Ref.~\cite{Wan:2026}; the physical cocycle here then fixes all mixed coefficients without additional angular-derivative moment sequences.

\paragraph{Conditional spectral interpretation.}
If the axis additionally has an unsubtracted physical representation with justified coefficient extraction $F(s)=\int d\sigma(E)/(E-s)$, the weighted inverse-variable relation is $\int h(x)\,d\mu(x)=\int h(1/E)\,d\sigma(E)/E$. The identified measure then gives atoms at $E_m=m/\alpha$ with weights $\alpha/m$. The spectral measure can have infinite total mass; the required weighted mass $\int d\sigma(E)/E=a_0$ is finite. Here $E$ is the variable with the same dimension as $s$, namely mass squared in the stated convention. Interpretation as a spectrum of coupled physical states is conditional on the physical spectral representation, and determines neither spin multiplicities nor uncoupled states.

\section{Assumptions, separation examples, and reproducibility}
\label{sm:physical-scope}

The physical necessity theorem uses the four-dimensional supersymmetric
Ward equations, unbroken $SU(4)$, scalar parity, standard massless planar
factorization with nonzero SYM coupling, and local contact terms. The
positive axis measure is required only for the amplitude uniqueness
corollary. Neither theorem uses an auxiliary identity model as a premise.
The present Supplement follows only the physical six-point route.

\paragraph{Relation to the three-channel master equation.}
The gauge-theory equation of Ref.~\cite{Huang:2026infrared}, Eq.~(2), follows directly from the cocycle and normalized kernel proved above. For $x_1+x_2+x_3+x_4=0$, set $s_{ij}=x_i+x_j$, $D(z)=\gamma(z)\gamma(-z)$ and $G_\gamma=\prod_{i=1}^4\gamma(x_i)$. Define
\begin{align*}
\mathcal E={}&(x_1x_3-x_2x_4)B(x_1,x_2)B(x_3,x_4)\\
&+x_2x_4B(x_1,s_{23})B(x_3,s_{14})-x_1x_3B(x_2,s_{14})B(s_{23},x_4).
\end{align*}
Substitution of $B(s,t)=\gamma(s)\gamma(t)/\gamma(s+t)$ and $x_1x_3-x_2x_4=s_{12}s_{23}$ gives
\[
 \mathcal E=G_\gamma D(s_{23})\bigl[q(s_{12})q(s_{23})+q(x_2)q(x_4)-q(x_1)q(x_3)\bigr]=0.
\]
Indeed, the bracket is $q(a+b)q(b+c)-q(a)q(c)-q(b)q(a+b+c)$ with $(a,b,c)=(x_1,x_2,x_3)$. The addition law reduces it to $q(a)q(c)[q'(b)^2+\kappa q(b)^2-1]$, which vanishes by the first integral, with $\kappa=6c_{00}$. All divisions above are by formal units $\gamma$ and $D$; $q(z)=z/D(z)$ introduces no division by a kinematic sum. The identity therefore includes vanishing sums, $\kappa=0$, and arbitrary formal odd cochain data. This provides a direct comparison of the four-point constraints; the differential contact-elimination certificate is the construction developed here.

\begin{center}
\small
\begin{tabular}{p{0.29\linewidth}p{0.65\linewidth}}
\toprule
Input or conclusion & Role and boundary \\
\midrule
Scalar parity & Allows the two parity-related Ward rows to act on the same
scalar vector. Supersymmetry alone does not impose it. \\
Planar local factorization & Makes the subtracted overlap remainder regular
and forbids repairing a Ward obstruction with a nonplanar pole. \\
Complete five-point freedom & Its arbitrary $T,P$ generators are retained
before the boundary functional is applied. \\
Gram relations & Both the scalar relation and all five epsilon relations
are included before differentiating the polynomial equations. \\
Common positive axis measure & Fixes the UV moments, chooses the nonzero
trigonometric branch, and eliminates all remaining odd cochain data. \\
Convergent four-point germ & The common positive measure forces convergence of the formal series. Identification with a physical continuation uses a common connected meromorphic domain. \\
Higher-point completion & The constraints are necessary; sufficiency for
arbitrary surviving four-point data and uniqueness at all multiplicities
are not proved. \\
\bottomrule
\end{tabular}
\end{center}

Two exact examples separate the inputs. First, take the cochain
$\gamma(z)=\sqrt{\lambda z/\sin(\lambda z)}$ with $\lambda>0$.
It obeys the cocycle and sine-kernel conditions but has axis coefficients
$a_1=0$ and $a_2=2a_0^2/5>0$. A nonnegative measure on $[0,\infty)$ with
zero first moment is supported at zero and cannot have this second moment.
This is a four-point solution of the two necessary conditions, not a
claimed physical six-point completion.

Second, let $\widetilde A=A_{\rm V}[1+\beta(stu)^2]$, where $\beta$ has
the compensating mass dimension. The factor is fully symmetric, and its
deviation on approaching the axis starts at $t^2$, so all monodromy
relations and the full pole-subtracted axis are unchanged. For numerical
units $\beta=1$, the factor $C(x,y)=1+[xy(x+y)]^2$ has cocycle defect
\[
\left[C(x,y)C(x+y,z)-C(y,z)C(x,y+z)\right]_{(x,y,z)=(1,1,2)}=-4080.
\]
It therefore fails the physical cocycle. The example imposes neither
fixed-angle unsubtracted dispersion nor all partial-wave positivity; it
shows precisely why the single axis measure suffices only after the
stronger cocycle has been derived.

\subsection{Exact certificates and independent checks}
The code and inputs are openly available in the versioned Zenodo archive~\cite{Sheng:WardCode}. The polynomial and rational certificates act on arbitrary function or multiplier data. They verify the physical scalar projections, both chiral sewings and the mixed determinant, the Ward rows and full Gram presentation, the complete five-point generators, the contact annihilator and full $T,P$ response, the arbitrary-function linear source, all three quadratic cut contributions, and the reconstruction of all 72 matrix entries. Spinor samples and decimal moment intervals are separately classified as cross-checks.

The finite replay at four-point degrees $2,4,6,8,10,12$ is an additional
check of the closed functional formula. The proof of arbitrary EFT order
is the polynomial annihilator, complete five-point response, and
arbitrary-function source identity, together with formal coefficient
extraction. These computations provide reproducible evidence; external
independent verification of the physical reduction remains a separate
scientific review step.

\paragraph{Additional independent algebraic audits.}
The verification record separates all-order mathematical arguments from finite-order checks. The physical-chart and five-point certificates verify the nonzero Jacobian witnesses, exact Gram identities, both chiral sewings, and complete monomial generators. The independent functional audit checks all 660 product-rule coefficients in nine variables, the compact 48-entry certificate, and the complete $T,P$ response. The independent source audit uses the two-parameter resolvent to span every symmetric polynomial and verifies exact division before curve restriction. The cocycle--kernel audit checks the physical path, cut coefficients, and arbitrary-index recurrence identities. The moment audit checks exact normalization and tail identities and separately cross-checks seven finite intervals against a defining-series enclosure. None of the finite samples supplies the inference to arbitrary EFT order. Reproducible algebraic identities and the analytic proofs in this Supplement have distinct roles.

\end{document}